\RequirePackage{fix-cm}
\documentclass{svjour3}                    
\smartqed  
\usepackage{graphicx}
\usepackage{bm,color,bbm}
\usepackage{mathtools}
\usepackage{subfigure}
\usepackage{centernot}
\usepackage{hyperref}

\newcommand{\beq}{\begin{equation}}
\newcommand{\eeq}{\end{equation}}
\newcommand{\bqa}{\begin{eqnarray}}
\newcommand{\eqa}{\end{eqnarray}}
\newcommand{\nn}{\nonumber}

\newcommand{\smallfrac}[2]{\mbox{$\frac{#1}{#2}$}}

\newcommand{\half}{\smallfrac{1}{2}}

\definecolor{maroon}{rgb}{0.7,0,0}

\definecolor{ngreen}{rgb}{0.3,0.7,0.3}

\definecolor{golden}{rgb}{0.8,0.6,0.1}

\journalname{Foundations}
\begin{document}

\title{Bell vs Bell: a ding--dong battle over quantum incompleteness
}

\author{Michael J. W. Hall        
}


\institute{ Department of Theoretical Physics \\Research School of Physics \\Australian National University \\Canberra ACT 0200, Australia \\
}

\date{}

\maketitle

\begin{abstract}
Does determinism (or even the incompleteness of quantum mechanics) follow from locality and perfect correlations?  In a 1964 paper  John Bell gave the first demonstration that quantum mechanics is incompatible with local hidden variables. Since then a vigorous debate has rung out over whether he relied on an assumption of determinism or instead, as he later claimed in a 1981 paper, derived determinism from assumptions of locality and perfect correlation. This paper aims to bring clarity to the debate via simple examples and rigorous results. It is first recalled, via quantum and classical counterexamples, that the weakest statistical form of locality consistent with Bell's 1964 paper (parameter independence) is insufficient for the derivation of determinism.
Attention is then turned to critically assess Bell's appealing to the Einstein-Rosen-Podolsky (EPR) incompleteness argument to support his claim. It is shown this argument is itself incomplete, via counterexamples that expose two logical gaps. Closing these gaps via a strong ``counterfactual'' reality criterion enables a rigorous derivation of both determinism and parameter independence, and in this sense justifies Bell's claim.   Conversely, however, it is noted that whereas the EPR argument requires a weaker ``measurement choice'' assumption than Bell's demonstration, it nevertheless leads to a similar incompatibility with quantum predictions rather than to quantum incompleteness.

\end{abstract}

\newpage

\section{Introduction}
\label{sec:intro}

\begin{quote}
	``{\it I have made an effort to present the deduction \dots 
		shorn of all superfluous mathematical technicalities and woolly interpretative commentary. (A reader as yet unfamiliar with the literature
		will be astounded to see the incredible metaphysical extravaganzas to
		which this subject has led.)}'' --- van Fraassen~\cite{fraassen}
\end{quote}

Bell's theorem, that there are no local hidden variable models that can reproduce all quantum correlations~\cite{bell1964,bellreview}, is one of the most surprising results of twentieth century physics. It has important ramifications for both physics and metaphysics, not only underlying the security of a number of device-independent communication protocols such as quantum key distribution and random number generation, but also ruling out naive classical interpretations of quantum probability~\cite{bellreview}.  

The mathematics required to prove Bell's theorem is remarkably straightforward.
However, there is one aspect of the assumptions used in Bell's first exposition that has generated differing opinions and vigorous debate~\cite{fraassen,wigner}--\cite{caval}:  do the assumptions of locality and perfect correlation made in his 1964 paper~\cite{bell1964} necessarily imply that certain measurement outcomes are predetermined, or must such determinism be postulated?

It may be noted that the answer to this question is unimportant in one sense: later generalisations of Bell's theorem do not rely on assuming perfect correlations~\cite{bellreview}   (with the further advantage that, unlike the 1964 result, they are experimentally testable~\cite{bell1981}; see also Section~\ref{sec:free}).    Nevertheless, the question remains of strong interest for several reasons.  First, there is no general consensus on the answer: while early commenters such as Wigner wrote that Bell postulated both deterministic hidden variables and locality~\cite{wigner} (see also~\cite{suppes,dem}), Bell himself later claimed that determinism was in fact inferred rather than assumed in the 1964 paper~\cite{bell1981}, and the debate has only intensified since then~\cite{fraassen,wutterich1}--\cite{caval}. Second, the validity of Bell's claim would suggest that there is no choice between giving up determinism or giving up locality in interpreting quantum phenomena: in the presence of perfect correlations it would be compulsory to give up locality~\cite{norsen2006,tumulka,shim83}---a conclusion that adds further fuel to the fire and is strongly contested in its own right~\cite{many}--\cite{griff}.  Finally, noting that Bell appeals to the famous Einstein-Podolsky-Rosen (EPR) argument for the incompleteness of quantum mechanics~\cite{epr} to support his claim, the debate puts the latter argument itself into question.

With the aim of making a clear and concise contribution to these issues, this paper is guided by the quote given at the beginning (taken from an excellent early discussion on the subject by van Fraassen~\cite{fraassen}), which continues to resonate decades later.
In particular, attention will be focused on what can be proved in a simple yet rigorous and general manner, and what can be disproved via simple counterexamples, whilst avoiding woolly assertions and metaphysical extravaganzas!

In Section~\ref{sec:loc} it is recalled that the weakest statisical sense of locality consistent with Bell's 1964 paper, parameter independence, is too weak for deriving determinism for perfectly correlated measurement outcomes (Proposition~\ref{prop1}). 
This still leaves open the possibility of an alternative route for inferring determinism, based on Bell's appeal to the EPR incompleteness argument to support his claim. However,  two gaps in the EPR logic are identified  via simple counterexamples in Section~\ref{sec:epr}, i.e., the argument itself  is incomplete, so that determinism and quantum incompleteness do not strictly follow (Propositions~\ref{prop2} and~\ref{prop3}). Fortunately, as shown in Section~\ref{sec:strong}, a suitable strengthening of the EPR reality criterion to a ``counterfactual reality'' criterion, combined with a ``accessible choice'' condition (weaker than the ``free choice'' condition required for Bell inequalities), closes these gaps and leads to rigorous derivations not only of determinism but also, in a closing of the circle, of parameter independence (Propositions~\ref{prop4} and \ref{prop5}).
Finally, it is noted in Section~\ref{sec:free} that the EPR argument, even in strengthened form, leads to contradictions with quantum predictions rather than to quantum incompleteness (Proposition~\ref{prop6}), as well as to a related class of experimentally testable inequalities (Proposition~\ref{prop7}).  
Conclusions are given in Section~\ref{con}.

\section{Statistical locality, perfect correlation and determinism}
\label{sec:loc}

\subsection{Statistical formulation of locality: parameter independence}
\label{sec:pi}

Bell spells out his intended sense of locality in several parts of his 1964 paper~\cite{bell1964}:
\begin{quote}
	``{It is the requirement of locality, or more precisely that the result of a measurement on one system be unaffected by operations on a distant system, with which it has interacted in the past \dots}''
\end{quote}
\begin{quote}
	``{Now we make the hypothesis, and it seems one at least worth considering, that if the two measurements are made at places remote from on another, the orientation of one magnet does not influence the result obtained with the other.}''
\end{quote}
\begin{quote}
	``{In a theory in which parameters are added to quantum mechanics to determine the results of individual measurements, without changing the statistical predictions, there must be a mechanism whereby the setting of one device can influence the reading of another instrument, however remote.}''
\end{quote}
These quotes strongly suggest a deterministic element is involved (since results must pre--exist before they can be ``affected'' or ``influenced'').  However, Bell later claimed that this element is a direct consequence of locality and perfect correlations~\cite{bell1981} (giving rise to the debate mentioned in Section~\ref{sec:intro}), and this claim will be examined in detail in this paper. For now we need only note a purely logical point: if determinism is able to be derived rather than assumed, then it must be possible to mathematically formulate Bell's sense of locality  {\it without} reference to determinism. Thus a statistical formulation of locality is required.

To obtain such a formulation, consistent with the above quotes, let $x$ and $y$ label possible measurements which may be made in two separate regions of spacetime, with respective outcomes labeled by $a$ and $b$ and joint probability density $p(a,b|x,y)$  for these outcomes. If $\lambda$ denotes any additional parameters of interest---arising, e.g., from a physical or mathematical model---then it follows from the basic rules of probability that
\beq \label{bayes}
p(a,b|x,y) = \sum_\lambda p(a,b,\lambda|x,y) = \sum_\lambda p(a,b|x,y,\lambda) \, p(\lambda|x,y).
\eeq
Here  summation is replaced by integration over any continuous range of $\lambda$. A natural definition of statistical locality for any such model can now be rigorously formulated.
\begin{definition}[Parameter independence] \label{def1}
	The probability distribution of the result of a measurement in one region is unaffected by measurement operations in a distant region, i.e.,
	\beq \label{weak}
	p(a|x,y,\lambda) = p(a|x,\lambda),\qquad p(b|x,y,\lambda) = p(b|y,\lambda),
	\eeq
	for all $x,y,\lambda$. 
\end{definition}
In particular, the first equality in Eq.~(\ref{weak}) implies that the setting $y$ of one measurement device cannot affect or influence the underlying statistics of the reading $a$ of a remote device, and vice versa for the second inequality.  This definition appears to be the weakest possible formulation of statistical locality that is consistent with the above quotes from Bell's 1964 paper (see also Proposition~\ref{prop5} below). It is, moreover, consistent with a supporting quote from Einstein given by Bell in a footnote to his 1964 paper~\cite{bell1964}:
\begin{quote}
	``{But on one supposition we should, in my opinion, absolutely hold fast: the real
		factual situation of system S2 is independent of what is done with the system
		S1, which is spatially separated from the former.}''
\end{quote}
The formulation in Definition~\ref{weak} has also been called  ``hidden locality''~\cite{fraassen} and ``locality''~\cite{jarrett}, with the now common term ``parameter independence'' being introduced by  Shimony~\cite{shimony}.

\subsection{Determinism does not follow from parameter independence and perfect correlations}
\label{sec:tooweak}

In 1981 Bell made the following claim about his 1964 paper~\cite{bell1981}:
\begin{quote}
	``{My own first paper on this subject \dots starts with a summary of the EPR argument {\it from locality  to} deterministic hidden variables. But the commentators have almost universally reported that it begins with deterministic hidden variables.}''
\end{quote}
\begin{quote}
	``{It was only in the context of perfect correlation (or anticorrelation) that
		determinism could be inferred [\dots] (for any indeterminism would have spoiled
		the correlation).}''
\end{quote}
Thus, it is claimed that determinism is inferred  as a logical consequence of locality and perfect correlation.   Note that ``determinism' here (and throughout this paper) refers to ``outcome determinism'', i.e., to {\it the prediction of the outcome of a measurement with certainty}, as is standard in the Bell inequality literature~\cite{bellreview}.    The ensuing debate in the literature arises in large part from the simple observation that the form of locality in Definition~\ref{def1} is inconsistent with such an inference, as may be demonstrated by simple counterexamples.
\begin{proposition} \label{prop1} If locality is formulated as parameter independence, then determinism does not follow from locality and perfect correlations, i.e., 
	\beq \label{eqprop1}
	{\rm parameter~independence~+~perfect~correlations} \centernot\implies {\rm determinism}.
	\eeq
\end{proposition}
\begin{proof} 
	A first counterexample of interest is a version of the position and momentum correlations considered in the EPR incompleteness argument~\cite{epr}.   In particular, suppose that $x_0\equiv Q_1$ and $x_1\equiv P_1$ correspond to measurements of conjugate position and momentum components $q_1$ and $p_1$ of a first particle, and $y_0\equiv Q_2$ and $y_1\equiv P_2$ to measurement of the corresponding components $q_2$ and $p_2$ of a second particle (where the particles may be classical or quantum), with corresponding measurement probability densities
	\beq \label{pqex1}
	p(q_1,q_2|Q_1,Q_2,\lambda) = \mu(q_1|\lambda)\delta(q_1+q_2),\qquad p(p_1,p_2|P_1,P_2,\lambda)=\nu(p_1|\lambda)\delta(p_1-p_2),
	\eeq
	\beq \label{pqex2}
	p(q_1,p_2|Q_1,P_2,\lambda) = \mu(q_1|\lambda)\nu(p_2|\lambda),\qquad p(p_1,q_2|P_1,P_2,\lambda)=\nu(p_1|\lambda)\mu(-q_2|\lambda),
	\eeq
	for some single-particle position and momentum probability densities  $\mu(q|\lambda)$ and $\nu(p|\lambda)$  conditioned on some underlying parameter $\lambda$.   Thus the positions and momenta are perfectly correlated (the sum of the position components is zero, as is the difference of the momentum components); parameter independence conditions~(\ref{weak}) are easily checked to hold; and yet the outcomes are not deterministic whenever $\mu$ and $\nu$ are not delta-functions. 
	
	A second noteworthy counterexample is  based on the perfect correlations of ``PR-boxes''~\cite{bellreview,rastall}. In particular, suppose that there are two possible measurements $x_0, x_1$ and $y_0,y_1$ in each region, each having outcomes $a,b=\pm1$, with $p(a,b|x_j,y_k,\lambda)=\half\delta_{ab,(-1)^{jk}}$ for all values of some underlying parameter $\lambda$.  Thus the outcomes are perfectly correlated for all choices of settings, with $b=(-1)^{jk}a$, and it is easily checked that parameter independence holds, with $p(a|x_j,y_k,\lambda)=\half=p(b|x_j,y_k,\lambda)$ in Eq~(\ref{weak}), despite the outcomes not being deterministic. 
\end{proof}

It is worth noting that both counterexamples in the proof of Proposition~\ref{prop1} (as well as the counterexample of standard quantum mechanics given by Wiseman~\cite{wiseman}) have underlying models (with additional underlying parameters) that satisfy both parameter independence and determinism. This is easily seen for the first counterexample, which is compatible with deterministic outcomes corresponding to members of a joint phase space ensemble with density $p(q_1,q_2,p_1,p_2|\lambda)=\mu(q_1|\lambda)\nu(p_1|\lambda)\delta(q_1+q_2)\delta(p_1-p_2)$ for each $\lambda$. More generally, local deterministic models exist for {\it any} set of correlations $\{p(a,b|x,y)\}$, including for the second counterexample in the proof, provided that the dependence of $p(\lambda|x,y)$ in Eq.~(\ref{bayes}) on $x$ and $y$ is sufficiently nontrivial~\cite{hallbrans}. 

However, Bell's claim is that determinism, not just the mere possibility of determinism, can be inferred from perfect correlations and locality. Proposition~\ref{prop1} immediately implies that this claim fails {\it if} parameter independence is used as a proxy for locality. Hence, while it appears to be the weakest form of statistical locality consistent with his 1964 paper, it is too weak to play the role of locality in his claim.

Fortunately, Bell's claim as quoted above suggests that Bell's intended form of locality is indeed stronger, and is directly related to the EPR incompleteness argument. Hence, if the claim can be shown to rely on a suitable formulation of locality via this argument, then it is logically valid. And this can indeed be shown --- but only if the EPR argument is suitably strengthened to close two gaps. This is the subject of the next two sections.

\section{Can the EPR argument be considered complete?}
\label{sec:epr}

As noted above, Bell's 1981 claim, to have inferred determinism from locality and perfect correlations in his 1964 paper, is based on a reference to the EPR incompleteness argument (see the quote at the beginning of Section~\ref{sec:tooweak}).   And indeed, in the 1964 paper he gives a very brief summary of this argument, in the context of two perfectly correlated spins $\bm \sigma_1$ and $\bm \sigma_2$, to conclude that~\cite{bell1964}
\begin{quote}
	{``Since we can predict in advance the result of measuring any chosen component of $\bm\sigma_2$, by previously measuring the same {\rm [perfectly correlated]} component of $\bm \sigma_1$, it follows that the result of any such measurement must be predetermined.}"
\end{quote}
Hence his claim is perfectly justified, providing that the EPR argument can in fact be used to derive determinism in this way.  But is this the case?

In this regard that Bell appears to have (e.g., in contrast to Niels Bohr~\cite{bohr}) fully accepted the EPR argument, and so does not spell out any details of how determinism rigorously follows from it. He is content for example, when revisiting and motivating the example of perfect spin correlations in Section~3 of his 1981 paper~\cite{bell1981}, to simply state that ``{if we do not accept the intervention on one side as a causal influence on the other, \it we seem obliged to admit} that the results on both sides are determined in advance'' (my italics).  Likewise, in earlier papers he similarly argues that  ``{\it This strongly suggests} that the outcomes of such measurements, along arbitrary directions, are actually determined in advance''~\cite{bellintro}; or asks  ``{\it Is it not more reasonable to assume} that the result was somehow predetermined all along?''~\cite{bellepr} (my italics). 

Hence, also noting the discussion in Section~\ref{sec:loc} and given the long--running debate mentioned in the Introduction, there is a need for a critical assessment of the EPR argument, to see whether it can indeed be used to give a rigorous derivation of determinism.  Such an assessment is the subject of this Section. It is found that there are in fact two logical gaps in the argument. These gaps will be closed in Section~\ref{sec:strong},  in part via a nontrivial strengthening of the EPR reality criterion, allowing both determinism and parameter dependence to be derived. This supports the essence of Bell's claim, as well as showing the consistency of parameter independence with the concept  of locality used in Bell's 1964 paper.

\subsection{Logic of the EPR argument}
\label{sec:logic}

The derivation of determinism via perfect correlation in the EPR paper begins with a simple sufficient condition for an element of reality, quoted here~\cite{epr}:
\begin{definition}[EPR reality criterion]  \label{epr}
	If, without in any way disturbing a system, we  can predict with certainty (i.e., with probability equal to unity) the value of a physical quantity, then there exists an element of physical reality corresponding to this physical quantity.
\end{definition}
This condition is then applied to derive determinism from  perfect correlations  (and in particular from perfect position and momentum  correlations), using the following logic:
\begin{itemize}
	\item[1.] Assume measurement $x_0$ is made in a first region, with result $a$ (assumption).
	\item[2.] The outcome of a measurement $y_0$ in a distant second region can then be predicted as $b=f(a)$ with certainty, for some 1:1 function $f$ (perfect correlation).
	\item[3.] This prediction can be obtained without disturbing the distant region in any way (assumption).
	\item[4.] Hence, the value of $b$ is an element of physical reality, prior to any actual measurement of $y_0$ in the distant region (EPR reality criterion).
\end{itemize}
The above steps are all that EPR explicitly use to derive a deterministic value, i.e., a predetermined real value of $b$, for a given perfect correlation. However, when one tries to make the above steps mathematically rigorous, an (easily fixed) ``asymmetry'' gap shows up in the above logic, as shown via a counterexample in Section~\ref{symm} below. Hence the EPR reality criterion is  insufficient for concluding the reality of $b$ via Steps 1--4.

Moreover, EPR then go on to consider what can be said if there are two or more perfectly correlated pairs of measurements, and add the following additional step:
\begin{itemize}
	\item[5.] If $x_1$ and $y_1$ are a second pair of perfectly correlated measurements, for the first and second regions respectively, then applying the same steps as above implies that the outcomes of both $y_0$ and $y_1$ are real and predetermined.
\end{itemize}
EPR rely on this last step to further conclude that quantum mechanics is   {\it incomplete}, in the sense that particular elements of reality, corresponding to the values of these predetermined outcomes, are not represented within quantum theory~\cite{epr}. In particular, they give an example where   the positions and momenta of each of two quantum particles are perfectly correlated, and yet have no deterministic values assigned to them   (see also the related example in Eqs.~(\ref{pqex1}) and~(\ref{pqex2}).    Bell similarly relies on Step~5 for the derivation of his 1964 inequality~\cite{bell1964} (since this inequality requires outcome determinism for multiple pairs of perfectly correlated measurements). However, as is shown via several counterexamples in Section~\ref{joint} below, Step~5 does not strictly follow as a logical consequence of Steps~1--4: there is a significant ``joint measurement'' gap. 

Two propositions follow immediately from these gaps and counterexamples:
\begin{proposition}\label{prop2} The logic of the EPR argument for the incompleteness of quantum mechanics is itself incomplete, i.e.,
	\beq \label{eqprop2}
	{\rm EPR~reality~criterion} \centernot \implies {\rm  incompleteness~of~QM}  .
	\eeq
\end{proposition}
\begin{proposition}  \label{prop3} The EPR reality criterion is insufficient to derive determinism for perfect correlations, i.e.,
	\beq
	{\rm EPR~reality~criterion} + {\rm perfect~correlations} \centernot \implies {\rm determinism}  .
	\eeq
\end{proposition}

Not all is lost, however. In particular, as previewed above, it is not difficult to formulate a stronger form of the reality criterion that (together with an ``accessible choice'' condition) allows the EPR argument to be completed and determinism to be derived for perfect correlations, as will be done in Section~\ref{sec:strong}.

\subsection{Two gaps in the logic}
\label{sec:gaps}

Propositions~\ref{prop2} and~\ref{prop3} above follow from counterexamples that highlight two gaps in the EPR logic, preventing Steps~4 and~5 from going through without additional assumptions. These gaps are explicitly identified and discussed below.

\subsubsection{The asymmetry gap}
\label{symm}

The first gap in the EPR logic is relatively minor, and arises from the inherent asymmetry of the argument: it requires only that a measurement made in the first region does not disturb the system in the second region in any way, as per Step~3 (e.g., ``at the time of measurement \dots no real change can take place in the second system in consequence of anything that may be done to the first system''~\cite{epr}). This does not rule out, however, the possibility that making a measurement in the second region  disturbs the system in the first region. Note this is consistent with no direct interaction between the two systems if it is the measurement device that is responsible for the disturbance. 

A class of counterexamples that exploits this gap is to allow (possibly superluminal) signalling from the second region to the first region, but not vice versa.
\begin{example} \label{ex1}
	Suppose that each region contains a single spin-$\half$ particle prepared in a totally mixed state, and that if a device measures spin in the $\bm y$ direction of the second particle, with result $b=\pm1$, it sends a signal (e.g., superluminally or along its backward lightcone) that puts the first particle into the $-b$ eigenstate of spin in the $\bm y$ direction. Hence, if a measurement of spin in the $\bm x$ direction is made in the first region, the statistics of the singlet state are reproduced, i.e., $p(a,b|\bm x,\bm y)=\frac14(1-ab\,\bm x\cdot\bm y)$. 
\end{example}

In particular, whenever both spins are measured in the same direction, i.e., $\bm x = \bm y$, then (i)~the outcome in the second region can be predicted with probability unity from the result of a measurement in the first region, and (ii)~there is no disturbance of the second region by any measurement carried out in the first region, thus fulfilling the conditions for the EPR reality criterion.  Nevertheless, contrary to Step~4 of the EPR logic, there is no pre-existing element of reality for the outcome of the measurement in the second region: a totally random and unpredictable result $b=\pm 1$ is obtained for any spin direction $\bm y$. Thus, use of the EPR reality criterion to derive a predetermined value for $b$ fails, providing a counterexample to the EPR logic, and Proposition~\ref{prop2} immediately follows. 


The obvious way to close the asymmetry gap is simply to symmetrise the EPR logic, by further assuming that measurements in the second region do not disturb the first region in any way (which has the additional advantage that one can also obtain the reality of the outcome of $x_0$, via Steps~1--4 with the roles of $x_0$ and $y_0$ reversed). Thus both $a$ and $b$ are predetermined {\it if} measurements of both $x_0$ and $y_0$ are made. It turns out, however, that closing the asymmetry gap in this way is not sufficient for completing the EPR argument: there is a second gap!

\subsubsection{The joint measurement gap}
\label{joint}

To expose the second (and most important) gap in the EPR logic (see also the analysis in Section~2 of~\cite{shimreview}), observe that it starts with the assumption in Step~1 that $x_0$ is measured.  Without this assumption the  value of $b$ cannot be predicted with certainty, as required for Steps~2 and~4. 
It follows that, even if the asymmetry gap is closed via symmetrisation as above,  
one cannot proceed directly from Step~4 to Step~5. In particular, if there is also a second pair of perfectly correlated measurements $x_1$ and $y_1$,  then it follows from Steps~1--4 that the outcomes of both $y_0$ and $y_1$ are real and predetermined  only if {\it both} $x_0$ and $x_1$ are measured without disturbing the predicted perfect correlations. Conversely, if they cannot be jointly measured in this way then Step~5 does not logically follow---there is a ``joint measurement'' gap.

There are various classes of counterexamples that exploit the joint measurement gap. In his response to the EPR paper, Bohr gives a quantum counterexample~\cite{bohr}. In particular, one cannot jointly measure the position and momentum of a first particle as they require physically incompatible experimental arrangements, and hence one cannot jointly infer the position and momentum of a distant second particle via the EPR logic.  The spin measurements considered by Bell~\cite{bell1964} supply a similar counterexample (since a Stern-Gerlach magnet cannot simultaneously have two measurement orientations). One can also construct   classical counterexamples, in which the required joint measurement of position and momentum is forbidden by an epistemic restriction~\cite{bart}. Proposition~\ref{prop3} immediately follows from these counterexamples.

EPR allude to the joint measurement gap in the penultimate paragraph of their paper, where they reject a possible replacement of their reality criterion  by a more restrictive criterion (in reference to the reality of the momentum $P$ and position $Q$ of a particle in the second region)~\cite{epr}:
\begin{quote}
	``Indeed, one would not arrive at our conclusion if one insisted that two or more physical quantities can be regarded as simultaneous elements of reality {\it only when they can be simultaneously measured or predicted.} \dots This makes the reality of $P$ and $Q$  depend upon the process of measurement carried out on the first system, which does not disturb the second system in any way. No reasonable definition of reality could be expected to permit this.''
\end{quote}
This quote suggests that a ``reasonable'' way to close the joint measurement gap is by requiring that the reality of a physical quantity in a given region is {\it prior} to any process of making a measurement that does not disturb that region (Clauser and Shimony call such independence of reality from measurement ``physical realism''~\cite{shimreview}). This is a {\it counterfactual} requirement, i.e., an element of reality is assumed to exist whether or not some measurement is made from which it can be inferred (see also Section~7.4.2 of~\cite{muynck} and Section~2 of~\cite{bruk}), and will be seen to allow Step~5  of the EPR logic in Section~\ref{sec:logic} to go through {\it if} it is possible to measure either one of $x_0$ and $x_1$ in the first region. The latter condition is essentially a ``measurement choice'' assumption (usually left implicit in discussions of the EPR argument), that rules out superdeterministic counterexamples in which there is only a single predetermined choice of measurement for any experiment, fixed by initial conditions in the far past.

\section{Recovering Bell's claim (and more) from a strengthened EPR argument}
\label{sec:strong}

\subsection{Closing the gaps with a stronger reality criterion and a measurement choice assumption}
\label{close}

It is clear from the counterexamples discussed in Section~\ref{sec:gaps} that the EPR logic is not sufficient to rigorously support Bell's claim or to derive the incompleteness of quantum mechanics, as per Propositions~\ref{prop2} and~\ref{prop3}: stronger assumptions are needed. It is further clear from the discussion of the counterexamples that natural choices for these further assumptions are a stronger (symmetrised) non--disturbance assumption to close the first gap, and a stronger (counterfactual) reality assumption and a measurement choice assumption to close the second gap. 

The first two assumptions are conveniently implemented by strengthening the EPR reality criterion in Definition~\ref{epr} as follows. 
\begin{definition}[Counterfactual reality criterion] \label{counter}
	If, without in any way disturbing {\textit{\textbf{or being disturbed by}}} a system, we  can predict with certainty (i.e., with probability equal to unity) the value of a physical quantity, then there exists, {\textit{\textbf{prior to us making the prediction}}}, an element of physical reality corresponding to this physical quantity. 
\end{definition}
Here the bolded phrases indicate additions to the EPR reality criterion in Definition~\ref{epr},   which are clearly necessary for closing the asymmetry and joint measurement gaps respectively in the light of the counterexamples in Section~\ref{sec:epr}.  
This leaves the third assumption to be implemented as an explicit condition:  
\begin{definition}[Accessible choice] \label{choice} 
	It is always possible to choose between alternative measurements.
\end{definition}  

This assumption, that no measurement is inaccessible to the choice of the experimenter,    is clearly necessary for closing the joint measurement gap   (see last paragraph of Section~\ref{joint}).   In particular, it ensures that different elements of physical reality, corresponding to making  different measurements, can be both be inferred via the EPR argument. 
It is not particularly controversial, and indeed Bohr emphasised in his reply to EPR that  ``we are \dots left with a {\it free choice} whether we wish to know the momentum of the particle or its initial position''~\cite{bohr}. 
Further, as will be discussed in Sec.~\ref{sec:free}, it is weaker than the ``measurement independence'' assumption required for deriving Bell inequalities.   

The counterfactual reality criterion and accessible choice assumption allow modified versions of Steps~4 and~5 of the EPR logic in Section~\ref{sec:logic} to go through, yielding positive counterparts to the no--go results in Propositions~\ref{prop1} and~\ref{prop3}, as shown below.

\subsection{Deriving determinism}
\label{sec:incom}

The stronger criterion of counterfactual reality in Definition~\ref{counter}, combined with the accessible choice assumption in Definition~\ref{choice}, restores the desired logical rigour of the EPR argument for the incompleteness of quantum mechanics. In particular, in contrast to the negative results in Propositions~\ref{prop1} and~\ref{prop3}:
\begin{proposition} \label{prop4} Deterministic outcomes, for pairs of perfectly correlated measurements in separated regions, follows from the counterfactual reality criterion and accessible choice of measurement, i.e.,
	\begin{align} 
		{\rm counterfactual~reality~criterion} +{\rm accessible~choice} &+ {\rm perfect~correlations} \nn\\
		&\qquad~~~~\implies   {\rm determinism} .
	\end{align}
\end{proposition}	
\begin{proof}
	The proposition is shown by modifying the EPR logic in Section~\ref{sec:logic}, for two pairs of perfectly correlated measurements $x_0,y_0$ and $x_1,y_1$, as follows (with primes and italics indicating modified steps): 
	\begin{itemize}
		\item[1.] Assume measurement $x_0$ is made in a first region, with result $a$ (assumption).
		\item[2.] The outcome of a measurement $y_0$ in a distant second region can then be predicted as $b=f(a)$ with certainty, for some 1:1 function $f$ (perfect correlation).
		\item[3$^\prime$.] This prediction can be obtained without disturbing {\it or being disturbed by} the distant region in any way (assumption).
		\item[4$^\prime$.] Hence, the value of $b$ is an element of physical reality, prior to any actual measurement of $y_0$ in the distant region, {\it and prior to making the prediction via an actual measurement of $x_0$} (counterfactual reality criterion).
		\item[5$^\prime$.] If $x_1$ and $y_1$ are a second pair of perfectly correlated measurements, for the first and second regions respectively, {\it and each of $x_0$ or $x_1$ are possible measurement choices in the first region}, then applying the same steps as above implies that the outcomes of both $y_0$ and $y_1$ are real and predetermined {\it prior to any actual measurement of $x_0$ or $x_1$} (accessible choice).
	\end{itemize}
\end{proof}

It is seen that Step~3$^\prime$ of the proof of Proposition~\ref{prop4} closes the asymmetry gap in Section~\ref{symm}, and Steps~4$^\prime$ and~5$^\prime$ close the joint measurement gap in Section~\ref{joint}. 
Unlike EPR, however, it will not be concluded here that the incompleteness of quantum mechanics follows, for the reasons discussed in Section~\ref{sec:free} further below.  

It is also seen that Step 4$^\prime$ of the modified EPR logic in the proof of Proposition~\ref{prop4} is sufficient to demonstrate that the outcome of $y_0$ is predetermined for two perfectly correlated measurements $x_0$ and $y_0$.  Step 5$^\prime$  then further demonstrates that predetermined values for the outcomes of {\it multiple} pairs of perfectly correlated measurements, as required for the derivation of Bell's 1964 inequality~\cite{bell1964}. 

Thus Proposition~\ref{prop4} supports Bell's claim, quoted at the beginning of Section~\ref{sec:tooweak}, that determinism can be inferred via the EPR argument --- with the caveat that a strengthened form of the argument is needed, as above, to close the identified logical gaps.

\subsection{Deriving parameter independence}

With some additional work one can also obtain a rigorous derivation not only of determinism, as per Proposition~\ref{prop4}, but also of parameter independence, i.e., of the weak form of statistical locality in Definition~\ref{weak}. This is the content of Proposition~\ref{prop5} below.  First, however, a simple lemma is required.

\begin{lemma} 
	If the conditional distribution of a random variable $u$  is deterministic for given prior information $v$ (i.e., $p(u|v)=0$ for all but one value of $u$), then it remains deterministic when conditioned on any further information $w$ compatible with $v$ (i.e., $p(u|v,w)=0$ for all but one value of $u$).
\end{lemma} 
\begin{proof}
	We have $p(u|v)=0$ for all $u\neq u_0$, for some $u_0$ (determinism), and $p(w|v)>0$ (compatibility). Hence $0\leq p(u,w|v)\leq p(u|v)=0$ for $u\neq u_0$, and so, using the standard rules of probability, $p(u|v,w)=p(u,w|v)/p(w|v)=0$ for all $u\neq u_0$.
\end{proof}

This lemma is used several times in the proof of the following.

\begin{proposition} \label{prop5} Parameter independence for the outcomes of any sets of possible measurements $\{x_r\}$ and $\{y_s\}$, in respective separated regions, follows from the counterfactual reality criterion and the accessible choice condition if  each pair of measurements $(x_s,y_s)$  is perfectly correlated, i.e.,
	\begin{align} 
		{\rm counterfactual~reality~criterion} +{\rm accessible~choice} &+ {\rm perfect~correlations} \nn\\
		&\implies   {\rm parameter~ independence} .
	\end{align}
\end{proposition}
\begin{proof}
	Note first  that, if no assumptions are made, the probability  of outcome $b$ of some measurement $y_0$ could conceivably depend not only on $y_0$, but also on a distant measurement $x_0$ and its outcome $a$, as well as on some set of underlying variables $\lambda'_0$. Hence, in the most general case, this probability has the form  $p(b|a,x_0,y_0,\lambda'_0)$. Now, for the case of two perfectly correlated measurements $x_0$ and $y_0$, it follows from Proposition~\ref{prop4} (and in particular from the counterfactual reality criterion as per Step~4$^\prime$ in the proof of that Proposition),  that the outcome of $y_0$ is predetermined, and hence that $p(b|a,x_0,y_0,\lambda'_0)=0$ for all but the predetermined value of $b$, {\it prior} to making a measurement of $x_0$ to obtain some outcome $a$. Further, from Step~3$^\prime$ in the proof of Proposition~\ref{prop4}, actually making a measurement of $x_0$ to obtain some outcome $a$ in the first region does not disturb  the predetermined value of $b$ in the second region, i.e., $x_0$ and $a$ are redundant for determining the value of $b$. Hence $p(b|a,x_0,y_0,\lambda'_0)$ = $p(b|y_0,\lambda'_0) = 0$
	for all but one value of $b$.
	
	Second, suppose that some other measurement, $x_1$ say, rather than $x_0$, is actually then made in the first region. Then, since we have $p(b|y_0,\lambda'_0) = 0$ for all but one value of $b$, and noting $x_1$ must be compatible with $y_0$ and $\lambda_0'$ if it is possible to make it (as guaranteed by the accessible choice condition), it follows from the above Lemma that
	\beq p(b|x_1,y_0,\lambda'_0)=p(b|y_0,\lambda'_0) = 0
	\eeq
	for all but one value of $b$. Interchanging the roles of the first and second regions in the argument thus far, and letting $\lambda''_0$ denote any additional set of underlying variables needed to determine the outcome of $x_0$, it similarly follows that
	\beq p(a|x_0,y_1,\lambda''_0) = p(a|x_0,\lambda''_0) = 0 
	\eeq
	for all but one value of $a$, for any measurement $y_1$ that can be made in the second region. 
	
	Third, suppose there is some set of perfectly correlated pairs of measurements, $\{(x_s,y_s):s\in S\}$, for some index set $S$. The above two equations then generalise to, replacing subscripts 0 and 1 by $r$ and $s$,
	\beq \nn p(b|x_r,y_s,\lambda'_s) = p(b|y_s,\lambda'_s) =0,\qquad p(a|x_s,y_r,\lambda''_s) = p(a|x_r,\lambda''_s) =0 ,
	\eeq
	for all but one value of $a$ and $b$ and all $r,s\in S$.
	Finally, letting $\lambda$ denote the set of pairs $\{(\lambda'_s,\lambda''_s):s\in S\}$, and swapping $r$ and $s$ in the second equation, application of the above Lemma to each of the distributions gives
	\beq p(a|x_r,y_s,\lambda)=p(a|x_r,\lambda)=0,\qquad p(b|x_r,y_s,\lambda) = p(b|y_s,\lambda) = 0,
	\eeq
	for all but one value of $a$ and $b$ and for all $r,s\in S$. Thus, the outcomes of {\it all} measurement pairs $(x_r,y_s)$ are deterministic, and satisfy parameter independence as per Eq.~(\ref{weak}), as required.
\end{proof}

Proposition~\ref{prop5}, in  demonstrating that both determinism and parameter independence follows from the strengthened EPR argument, shows that the latter notion of locality is consistent with that used in Bell's 1964 paper.  However, this closing of the circle  will be seen in the next section to come at a cost for the EPR argument: it can no longer be used to support the incompleteness of quantum mechanics.

\section{The EPR argument clashes with quantum predictions, not quantum completeness}
\label{sec:free}

The original aim of the EPR argument was to deduce the incompleteness of quantum theory from a simple reality criterion (Definition~\ref{epr}). As noted in Section~\ref{sec:logic}, the EPR logic was to argue this criterion leads to determinism for the outcomes of perfectly correlated pairs of position and momentum measurements, and observe that quantum mechanics fails to assign any such deterministic values. Thus, it appeared reasonable to conclude that ``the quantum-mechanical description given by wave functions is not complete''~\cite{epr}. There also appeared to be no fundamental issue standing in the way of a complete description, compatible with their reality criterion, leading EPR to write that ``\dots we left open the question of whether or not such a description exists. We believe, however, that such a description is possible''~\cite{epr}.

However, this belief has since turned out to be incompatible with the EPR argument itself. In particular, various results in the literature have established that the argument is in direct conflict with several predictions of quantum mechanics~\cite{redhead,ghzshim,mermin1990,conway} (see also the forceful exposition by \.{Z}ukowski and Brukner~\cite{bruk} in this regard).  Thus, there is no completion of quantum theory of the sort advocated by EPR.

A general version of the above result is formulated in Proposition~\ref{prop6} below, upgrading Proposition~\ref{prop2}. This version is also helpful in clarifying an important distinction between the roles and relative strengths of ``accessible choice'' in the EPR argument versus ``free choice'' in the derivation of Bell inequalities. The related Proposition~\ref{prop7} given below involves a similar distinction, and emphasises the relevance of a recent class of experimentally testable inequalities for ``measurement dependent local'' models of correlations~\cite{putz1,putz2,putz3,brunner2023}.

\begin{proposition} \label{prop6} The assumptions underpinning the EPR argument, when strengthened to avoid the logical gaps noted in Section~\ref{sec:epr}, are theoretically inconsistent with quantum mechanics, i.e., 
	\begin{align} 
		{\rm counterfactual}&~{\rm reality~criterion} +{\rm accessible~choice}  \nn\\
		&~~\qquad\implies   {\rm ~mathematical~inconsistency~with~quantum~predictions} .
	\end{align}
\end{proposition}

\begin{proposition} \label{prop7} The assumptions of determinism, parameter independence and accessible choice are experimentally inconsistent with quantum mechanics, i.e, 
	\begin{align} 
		{\rm determinism} &+  {\rm parameter}~{\rm dependence} + {\rm accessible~choice}  \nn\\
		&\implies   {\rm testable~inconsistency~with~quantum~predictions} .
	\end{align}
\end{proposition}

These propositions are demonstrated by example further below. Note that Proposition~\ref{prop6} immediately implies that a completion of quantum theory, of the sort envisaged by EPR, is impossible without giving up one or both of the counterfactual reality criterion or accessible choice (or, alternatively, modifying the predictions of the current theory). This is a stronger result than Bell's original 1964 inequality, which relies on a stronger choice assumption (see below). Proposition~\ref{prop7} is similarly stronger than experimentally testable Bell inequalities, such as the Clauser-Horne-Shimony-Holt (CHSH) Bell inequality~\cite{chsh}, which again rely on a stronger measurement choice assumption.

Examples demonstrating Proposition \ref{prop6} generally fall into the class of what may be termed EPR-Kochen-Specker theorems~\cite{relaxed}, applying to various sets of perfect correlations between separated quantum systems~\cite{redhead,ghzshim,mermin1990,conway}. For example, note that quantum mechanics predicts that equal spin directions of the bipartite state $3^{-1/2}(|1,-1\rangle+|-1,1\rangle-|0,0\rangle)$, corresponding to two spin-1 particles having zero total spin, are perfectly correlated~\cite{redhead,conway}. Hence, under the assumptions of Proposition~\ref{prop6}, determinism and parameter independence directly follow via Propositions~\ref{prop4} and~\ref{prop5} respectively. However, the predicted quantum correlations for this state are incompatible with the combination of determinism, parameter independence and accessible choice~\cite{redhead,conway,landsman}, and Proposition~\ref{prop6} immediately follows.
A similar contradiction may be obtained for the GHZ state of three spin-1/2 particles~\cite{ghzshim,mermin1990}, providing that the counterfactual reality criterion is generalized to perfect correlations between three distant regions. Note that the accessible choice assumption cannot be dropped in Proposition~\ref{prop6}, as it is a simple matter to construct models for such correlations that satisfy both determinism and parameter independence~\cite{relaxed}.

To verify that Proposition~\ref{prop6} is stronger than Bell's 1964 result, note that the latter relies not only on determinism and parameter dependence (which follow for perfect correlations via Propositions~\ref{prop4} and~\ref{prop5}), but also on the strong ``free choice'' assumption that measurement choices in each region are completely independent of system properties. In the context of hidden variables this latter assumption corresponds to requiring that
\beq \label{measind}
p(\lambda|x,y) = p(\lambda)
\eeq
in Eq.~(\ref{bayes}), also known as ``measurement independence'', and is required to derive Bell inequalities more generally~\cite{bellreview,relaxed}. In contrast, Proposition~\ref{prop6} only requires the accessible choice assumption in Definition~\ref{choice}, i.e., that no choice of measurement is impossible, which corresponds to
\beq \label{acc}
p(x,y|\lambda) >0 ~~{\rm for}~~ p(\lambda) >0
\eeq
in the hidden variable context. It may alternatively be termed ``measurement accessibility'' in this context, and is clearly mathematically weaker than measurement independence (noting it is equivalent to  $p(\lambda|x,y)=p(x,y|\lambda)p(\lambda)/p(x,y)>0$ for $p(\lambda)>0$).  This distinction between the degrees of measurement choice required for EPR-Kochen-Specker theorems and Bell inequalities has been previously noted by Landsman (following Eq. (C10) of~\cite{landsman}).

Finally, as noted in the Introduction,  perfect correlations are not experimentally testable~\cite{bell1981}, ruling out direct tests of the 1964 Bell inequality and the 1935 EPR argument (in contrast to later Bell inequalities~\cite{bellreview} and generalizations of EPR correlations~\cite{reid,steering}).
However, the examples demonstrating Proposition~\ref{prop7} do not rely on perfect correlations, but rather on the combination of determinism, parameter independence and accessible choice, making an experimental test of this combination possible. In particular, a class of testable inequalities for measurement dependent local models of correlations have been recently obtained by P\"utz and Gisin and others~\cite{putz1,putz2,putz3,brunner2023}, which hold for all models of correlations of the form
\beq
p(a,b,x,y) = p(a,b|x,y) p(x,y) = \sum_\lambda p(\lambda) p(x,y|\lambda) p(a|x,\lambda) p(b|y,\lambda),
\eeq
\beq
p(x,y|\lambda) \geq \ell >0,
\eeq
for some $\ell$. Such correlations clearly satisfy parameter dependence as per Eq.~(\ref{weak}) and accessible choice as per Eq.~(\ref{acc}), and include deterministic models (corresponding to  $p(a|x,\lambda), p(b|y,\lambda)\in\{0,1\}$) as a special case. Further, these inequalities are inconsistent with predicted quantum correlations for any positive value of $\ell$~\cite{putz1,putz3,brunner2023} (as has been verified by experiment for $\ell\geq0.090$~\cite{putz2}), and  Proposition~\ref{prop7} immediately follows.

	\section{Conclusions}
	\label{con}

	Bell's 1981 claim that determinism is inferred from locality and perfect correlations in his 1964 paper is seen to be justified, to a significant extent, by his appeal to the EPR argument for the incompleteness of quantum mechanics. In particular, while the EPR argument has some logical gaps in its original form, a strengthened argument based on additional assumptions leads to a rigorous derivation of determinism as per Proposition~\ref{prop4}. This provides a solid foundation for a consensus position  with respect to the Bell vs Bell debate referred to in the Introduction, that may help reconcile the two main camps to some degree (see~\cite{wiseman} for a campanology analysis of the debate).
	
	In particular, while the derivation supports the essence of Bell's claim quoted at the beginning of Section~\ref{sec:tooweak}, it nevertheless relies on the use of stronger assumptions than in the original EPR argument --- the counterfactual reality criterion and accessible choice condition in Section~\ref{close} --- that go beyond the general concept of locality. It follows that the predicted violation of the 1964 Bell inequality by perfect quantum correlations does not imply that the only option for explaining these correlations is to give up locality, as sometimes asserted~\cite{norsen2006,tumulka,shim83}.
	
	In addition to determinism,  the strengthened EPR argument also provides a rigorous derivation of parameter independence for perfect correlations, as per Proposition~\ref{prop5}, demonstrating that the argument is consistent with the weakest formulation of statistical locality consonant with Bell's 1964 paper.   Thus, despite being inconsistent with its original aim of demonstrating the incompleteness of quantum mechanics (Proposition~\ref{prop6}), the EPR logic remains of value in providing  a robust scaffolding for Bell's claim, as per Propositions~\ref{prop4} and~\ref{prop5}.  It is also of interest that the EPR argument relies on a weaker measurement choice assumption than is required for obtaining testable Bell inequalities, where this weaker assumption similarly leads to experimentally testable inequalities as per Proposition~\ref{prop7}.  
	
	Noting further that the combination of determinism and parameter dependence implies that the property of local causality, $p(a,b|x,y,\lambda)=p(a|x,\lambda)p(b|y,\lambda)$, is satisfied, it  follows from Propositions~\ref{prop4} and~\ref{prop5} that this property is also consistent with the strengthened EPR argument. This is of interest because, while local causality is equivalent to this combination for the case of perfect correlations (e.g.,~\cite{fraassen,norsen2006,maudlin,wiseman,tumulka}), it may also be used  (replacing the accessible choice condition~(\ref{acc}) with the stronger  free choice condition~(\ref{measind})) to derive general Bell inequalities that apply even in the absence of perfect correlations~\cite{bellreview,bell1981,chsh}, thus sidestepping the EPR argument altogether (see also Section~\ref{sec:free}).

	Finally, in addition to the Bell vs Bell debate referred to in the Introduction, there is a related debate over whether local causality is purely a locality concept or, in contrast, incorporates elements not directly related to locality (e.g., \cite{shim83,gisin,hallcom,bruk,jarrett,kronz,butter,norsen09}).  Since local causality is equivalent to the  combination of parameter independence  and the outcome independence condition $p(a,b|x,y,\lambda)=p(a|x,y,\lambda)p(b|x,y,\lambda)$~\cite{jarrett,shimony}, where the former is generally acknowledged to be an unambiguous locality concept (see Section~\ref{sec:pi}), this debate essentially boils down to whether the same can be said of outcome independence. While the debate is outside the scope of this paper, it is of interest to note here that outcome independence is far stronger than parameter independence with respect to the main focus of this paper, i.e., to  the derivation of determinism via perfect correlations. In particular, in contrast to the no--go results in Propositions~\ref{prop1} and \ref{prop3}, it may be shown that outcome independence and perfect correlations, by themselves, are sufficient to derive determinism~\cite{fraassen,halloutcome}.

\begin{acknowledgements}
I am grateful to Marek \.{Z}ukowski for pointing out a flaw in the original statement of Proposition~\ref{prop4} and for motivating Section~\ref{sec:free}, and to the anonymous referees for useful suggestions.
\end{acknowledgements}


\end{document}